\documentclass[10pt]{article}
\usepackage[top=1.25in, bottom=1.5in, left=1.5in, right=1.5in]{geometry}

\usepackage{amsmath,amsfonts,amssymb}
\usepackage{authblk}
\usepackage{booktabs}
\usepackage{caption}
\usepackage{enumitem}
\usepackage{hyperref}
\usepackage{listings}
\usepackage{microtype}
\usepackage{nicefrac}
\usepackage{stmaryrd}
\usepackage{subcaption}
\usepackage{xspace}

\usepackage{cleveref}

\usepackage[textsize=footnotesize,textwidth=1.5in]{todonotes}

\newcommand{\eff}{{\ensuremath{\mathbb{F}}}} 
\newcommand{\A}{\ensuremath{\mathcal{A}}}
\newcommand{\zed}{{\ensuremath{\mathbb{Z}}}} 
\newcommand{\D}{\ensuremath{\mathcal{D}}}

\usepackage{amsthm}
\newtheorem{theorem}{Theorem}[section]
\newtheorem{lemma}[theorem]{Lemma}
\newtheorem{corollary}[theorem]{Corollary}
\newtheorem{remark}{Remark}[section]
\newtheorem{example}{Example}[section]
\newtheorem{definition}{Definition}[section]
\newtheorem{constructionx}{Construction}[section]
\newenvironment{construction}
  {\pushQED{\qed}\constructionx}
  {\popQED\endconstructionx}
\let\oldconstruction\construction
\renewcommand{\construction}{\oldconstruction\normalfont}

\title{Unconditionally Secure  Non-malleable Secret Sharing and Circular External Difference Families}
\author[1]{Shannon Veitch}
\affil[1]{Department of Computer Science, ETH Zurich, Zurich, Switzerland}
\author[2,3]{Douglas R.\ Stinson\thanks{D.R.\ Stinson's research is supported by  NSERC discovery grant RGPIN-03882.}}
\affil[2]{School of Mathematics and Statistics\\
Carleton University\\
Ottawa, Ontario, K1S 5B6, Canada}
\affil[3]{David R.\ Cheriton School of Computer Science\\University of Waterloo\\ Waterloo ON, N2L 3G1\\Canada}

\date{\today}

\begin{document}

\maketitle

\begin{abstract}
Various notions of non-malleable secret sharing schemes have been considered. In this paper, we review the existing work on non-malleable secret sharing and suggest a novel game-based definition. We provide a new construction of an unconditionally secure non-malleable threshold scheme with respect to a specified relation. To do so, we introduce a new type of algebraic manipulation detection (AMD) code and construct examples of new variations of external difference families, which are of independent combinatorial interest.

\end{abstract}

\section{Introduction}


The concept of non-malleability in cryptography was introduced by Dolev et al.~\cite{SJC:DolDwoNao00} in the context of public-key encryption schemes. It ensures that an adversary cannot tamper with a ciphertext in a ``meaningful'' way. The idea was later adapted to commitment schemes \cite{STOC:DamGro03,JC:FisFis09} and codes \cite{ACM:DPW18}. 

In the context of secret sharing schemes, various notions of non-malleability have been considered. In this paper, we review the existing work on non-malleable secret sharing and suggest a novel game-based definition. The basic idea of our approach is that an adversary should not be able to modify shares in such a way that the reconstructed secret is related to the real secret in a pre-specified manner, according to a particular relation defined on the set of possible secrets.  

We provide a new construction of an unconditionally secure non-malleable threshold scheme with respect to additive  relations. To do so, we introduce a new type of algebraic manipulation detection (AMD) code and construct examples of new variations of external difference families (namely, circular external difference families), which are of independent combinatorial interest.

We begin by presenting background on secret sharing (threshold schemes) and robust secret sharing schemes. 

\subsection{Secret sharing}

A secret sharing scheme is a cryptographic primitive for splitting a secret into some shares and distributing the shares amongst a set of participants such that only certain authorized subsets of participants can reconstruct the secret. In a $(k,n)$-threshold scheme, there are $n$ participants and an authorized subset is any subset of size at least $k$. Constructions for threshold schemes were introduced independently in 1979 by Shamir \cite{ACM:Sha79} and Blakley \cite{MARK:Bla79}. Blakley's construction relies on finite geometries while Shamir's scheme uses polynomial interpolation to reconstruct secrets. 

Secret sharing schemes are mainly considered in an \emph{unconditionally secure} (or  \emph{information-theoretic}) setting. This setting ensures that the security guarantees hold regardless of the computational capabilities of the adversary. All schemes considered in this paper are assumed to be unconditionally secure.

The goals of unconditionally secure secret sharing schemes are twofold: correctness and perfect privacy. 
These are defined more formally as follows.

\begin{definition}[$(k,n)$-threshold scheme]
Let $1 < k \leq n$ and let $\mathcal{S}$ be the set of possible \emph{secrets}. There are $n$ participants in the scheme, denoted $P_1, \dots , P_n$, as well as an additional participant called the  \emph{dealer}.

In a $(k,n)$-threshold scheme, a {secret} $s \in \mathcal{S}$ is chosen by the  dealer. The dealer then constructs $n$ \emph{shares}, which we denote by $s_1, \dots , s_n$. The share $s_i$ is given to participant $P_i$, for $1 \leq i \leq n$.  

The following two properties should be satisfied.
\begin{description}
	\item[Correctness:] Any set of $k$ participants can reconstruct the secret from the shares that they hold collectively.
	\item[Perfect privacy:] No set of $k-1$ or fewer participants can obtain any information  about the secret from the shares that they hold collectively.
\end{description}
\end{definition}

The prototypical example of an unconditionally secure threshold scheme is Shamir's scheme \cite{ACM:Sha79}. For future use, we present Shamir's scheme now.

\begin{constructionx}[Shamir's Threshold Scheme]
Suppose $\mathbb{F}_q$ is a finite field, where $q$ is a prime or a prime power. The $(k,n)$-threshold scheme will  share a secret $s \in \mathbb{F}_q$, where $q \geq n+1$. 
\begin{description}
	\item[Share:] The dealer selects a random polynomial $f(x) \in \mathbb{F}_q[x]$ of degree $k-1$ such that $f(0) = s$.  Each share $s_i$ is an ordered pair, i.e., $s_i = (x_i, y_i)$, where the $x_i$’s are distinct and non-zero and $y_i = f(i)$. The $x_i$'s are public and the $y_i$'s are secret. The dealer gives share $s_i$ to participant $P_i$, for $i = 1,\dots,n$.
	\item[Recover:] Given $k$ shares, the participants use polynomial interpolation to reconstruct $f(x)$ and then they evaluate the polynomial $f(x)$ at $x=0$ to recover the secret $s$.
\end{description}
\end{constructionx}

In Shamir's scheme, we can use  polynomial interpolation to reconstruct the secret. In particular, we can use the \emph{Lagrange interpolation formula} to reconstruct the polynomial $f(x)$. Suppose we are working in a field $\mathbb{F}_q$ and are given $k$ (not necessarily distinct) elements in $\mathbb{F}_q$, say $y_1,\dots,y_k$. Let $x_1,\dots,x_k$ be distinct elements in $\mathbb{F}_q$. Then there is a unique polynomial $f(x) \in \mathbb{F}_q[x]$ with degree at most $k-1$ such that $f(x_i) = y_i$ for $1 \leq i \leq k$. The \emph{Lagrange interpolation formula} states that
\[
f(x) = \sum_{j=1}^k y_j \prod_{1 \leq h \leq k, h \neq j} \frac{x - x_h}{x_j - x_h}.
\]
When reconstructing secrets with Shamir's scheme, we are typically concerned with the evaluation of this polynomial at 0. In this case, it is sufficient to compute
\[
s = \sum_{j=1}^k y_j \prod_{1 \leq h \leq k, h \neq j} \frac{x_h}{x_h - x_j}.
\]
Now, if we define
\[
b_j = \prod_{1 \leq h \leq k, h \neq j} \frac{x_h}{x_h - x_j},
\]
for $1 \leq j \leq k$, then we can write $s = \sum_{j=1}^k b_jy_j$. The values $b_j$ are called the \emph{Lagrange coefficients} and they are publicly known values. The secret $s$ is just a linear combination of the $k$ shares.

\subsection{Robust secret sharing}
We review the notion of robust secret sharing due to Tompa and Woll \cite{JC:TomWol88}. This was the first paper to challenge the adversarial model of classical secret sharing. The work was motivated by the potential of malicious shareholders in a secret sharing scheme to cause an incorrect secret to be recovered by submitting incorrect shares. In the following scenario, we assume that the dealer is an honest participant, but the participants who receive shares may act maliciously. The game-based definition that follows is equivalent to the description given by Tompa and Woll \cite{JC:TomWol88}.

\begin{definition}[The Robustness Game]\label{robust.defn}
Assume a $(k,n)$-threshold scheme, where the secret $s$ is chosen equiprobably from the set $\mathcal{S}$ of possible secrets. Fix a non-negative integer $t$ such that $1 \leq t < k$. 

\begin{description}
\item[Step 1.] $t$ of the $n$ shares are given to the adversary.
The adversary modifies the $t$ shares to create new ``bad shares".

\item[Step 2.] A secret $s'$ is reconstructed using the $t$ ``bad shares" and $k-t$ of the original ``good shares". The adversary may choose which of the ``good shares" are used in reconstruction. The adversary wins the robustness game if the reconstructed secret $s'$ is a valid secret and $s' \neq s$.
\end{description}
Typically, we let $t = k-1$. For a positive real number $\epsilon <1$, if the adversary can only win this game with probability at most $\epsilon$ ($\epsilon$ is the \emph{cheating probability}), then we say that the threshold scheme is $\epsilon$-\emph{robust}.
\end{definition}

\begin{remark}
When we say the adversary ``may choose which of the `good shares' are used in reconstruction," this does not mean that the adversary gets to see the values of the shares. The intention is that the adversary is able to select some identifiers denoting which shares to use in reconstruction. For example, in Shamir's scheme, it is sufficient to allow the adversary to select certain $x$-coordinates corresponding to shares. Knowing these $x$-coordinates allows the adversary to use their knowledge of publicly known Lagrange coefficients to perform attacks that would otherwise be impossible if they had no control over which shares were used during reconstruction. This is demonstrated in \Cref{thm.tw1}.
\end{remark}


\begin{theorem}\label{thm.tw1} \cite{JC:TomWol88}
A $(k,n)$-Shamir threshold scheme is not $\epsilon$-robust for any $\epsilon < 1$.
\end{theorem}
\begin{proof}
We describe how the adversary can always carry out a successful attack by modifying only a single share. Suppose the scheme has secrets in
$\eff_q$ and suppose the shares $(x_1,y_1),\dots,(x_k,y_k)$ are used to recover the secret $s$.

The true secret, $s$, can be written as $s = \sum_{i=1}^k b_iy_i$ where the $b_i$'s denote the Lagrange coefficients corresponding to $x_1, \dots , x_k$ and the $y_i$'s are shares, for $i = 1,\dots,k$.

Assume that the adversary modifies a single share $y_1$ by adding some nonzero value $\delta$ to it. Then the adversary submits the incorrect share $y_1' = y_1 + \delta$. The reconstructed secret will be
\begin{eqnarray*}
s' &=& b_1y_1' + \sum_{i=2}^k b_iy_i \\ &=& \sum_{i=1}^k b_iy_i + b_1\delta\\
&=& s + b_1\delta.
\end{eqnarray*}

The reconstructed secret $s'$ is valid since there are no restrictions on the secrets in Shamir's original scheme. As long as $\delta \neq 0$, we have $s' \neq s$, so the adversary wins the robustness game.
\end{proof}

On the other hand, a robust threshold scheme can be achieved using various methods, the most general of which employ  \emph{AMD codes}. AMD codes were defined by defined by Cramer \emph{et al.} \cite{EC:CDFPW08}; additional information can be found in \cite{CFP,CPX}. 
Here we use the game-based definitions of AMD codes from \cite{DM:PatSti16}. There are two basic forms of AMD codes: \emph{weak}  and \emph{strong}; 
we first consider weak AMD codes. 

\begin{definition}[Weak AMD code]\label{def.amd}
Let $\mathcal{G}$ be an additive abelian group of order $n$ and let $\mathcal{A} = \{{A}_0,\dots,{A}_{m-1}\}$ be $m$ pairwise disjoint $\ell$-subsets of $\mathcal{G}$. Let $0 < \epsilon < 1$. Then $(\mathcal{G},\mathcal{A})$ is an \emph{$\epsilon$-secure weak $(n,m,\ell)$-AMD code} if an adversary cannot win the following \emph{AMD game} with  probability greater than $\epsilon$.
\begin{enumerate}
\item The adversary chooses a value $\Delta \in \mathcal{G} \setminus \{0\}$. 
\item The \emph{source} $i \in \{0,\dots,m-1\}$ is chosen uniformly at random.
\item The source is encoded by choosing $g$ uniformly at random from ${A}_i$.
\item The adversary wins if and only if $g + \Delta \in A_j$ for some $j \neq i$. 
\end{enumerate}
\end{definition}

\begin{constructionx}[Robust Secret Sharing Scheme] 
\label{robust.const}
We  use an $\epsilon$-secure weak $(n,m,\ell)$-AMD code defined over $\eff_q$ to construct an $\epsilon$-robust $(k,n)$-threshold scheme having secrets from $\mathcal{S} = \{0, \dots , m-1\}$.  Suppose that the secret $s \in \{0, \dots , m-1\}$ is chosen uniformly at random.

\begin{description}
\item[Share:] The dealer chooses an element $K \in A_s$ uniformly at random. The value $K$ is called the \emph{encoded secret}. Then the dealer computes shares for $K$ using the usual Shamir threshold scheme over $\mathbb{F}_q$.
\item[Reconstruct:] A set of $k$ players determine the encoded secret $K$ using polynomial interpolation. 
Then they determine the value $s$ such that $K \in A_s$. This value $s$ is the secret.
\end{description}
\end{constructionx}

We now sketch a proof that the above-described threshold scheme is $\epsilon$-\emph{robust}. 
Suppose the shares $y_1,\dots,y_k$ are used to reconstruct a polynomial and thereby determine the encoded secret, $K$.
We have $K = \sum_{i=1}^k b_iy_i$ where the $b_i$'s are the Lagrange coefficients. The shares $y_1,\dots,y_{k-1}$ are known to the adversary and can be modified. The share $y_k$ is not known to the adversary.

Assume that the adversary creates fake shares $y'_1, \dots ,y'_{k-1}$, where
$y'_i = y_i + \delta_i$ for $1 \leq i \leq k-1$. 
The reconstructed encoded secret will be
\begin{eqnarray*}
K' &=&  \sum_{i=1}^{k-1} b_iy'_i  + b_ky_k \\
&=&  \sum_{i=1}^{k-1} b_i(y_i + \delta_i)  + b_ky_k \\
&=&  K + \sum_{i=1}^{k-1} b_i\delta_i .
\end{eqnarray*}
In order for the adversary to win the robustness game, it must be the case that 
\[ \Delta = \sum_{i=1}^{k-1} b_i\delta_i \neq 0.\]
Define $\delta = \Delta (b_1)^{-1}$ (note that all the Lagrange coefficients are nonzero). Then the single modified share  
$y''_1 = y_1 + \delta$ yields the same reconstructed encoded secret as the above-described $k-1$ modified shares, because
\begin{eqnarray*}
b_1y''_1 + \sum_{i=2}^{k} b_iy_i  
&=&  \sum_{i=1}^{k} b_iy_i  + b_1 \delta  \\
&=&  K + \Delta\\
&=& K'.
\end{eqnarray*}

Hence, we can assume that the adversary only modifies a single share, say $y_1$. 
Further, this modification has the result that the encoded secret $K$ is modified to $K + \Delta$. The value $\Delta$ is chosen by the adversary, but $K$ is not known to the adversary. We are assuming that $K$ is determined by choosing a secret $s \in \{0, \dots , m-1\}$ uniformly at random and then encoding it by choosing $K$ uniformly at random from  $A_s$. So we are in the setting assumed by a weak AMD code. Hence the adversary wins the robustness game with probability at most $\epsilon$, because the AMD code is assumed to be $\epsilon$-secure.

\medskip

If there is  an arbitrary (nonuniform) distribution on the set of secrets, we can instead use a strong AMD code. In a strong AMD code, the adversary knows the source (but not the encoded source) before they choose $\Delta$.

\begin{definition}[Strong AMD code]
Let $\mathcal{G}$ be an additive abelian group of order $n$ and let $\mathcal{A} = \{{A}_0,\dots,{A}_{m-1}\}$ be $m$ pairwise disjoint $\ell$-subsets of $\mathcal{G}$. Let $0 < \epsilon < 1$. Then $(\mathcal{G},\mathcal{A})$ is an \emph{$\epsilon$-secure strong $(n,m,\ell)$-AMD code} if an adversary cannot win the following \emph{strong AMD game}  with  probability greater than $\epsilon$.
\begin{enumerate}
\item The \emph{source} $i \in \{0,\dots,m-1\}$ is specified and given to the adversary.
\item The adversary chooses a value $\Delta \in \mathcal{G} \setminus \{0\}$. 
\item The source is encoded by choosing $g$ uniformly at random from ${A}_i$.
\item The adversary wins if and only if $g + \Delta \in A_j$ for some $j \neq i$. 
\end{enumerate}
\end{definition}

It can be shown that if we modify Construction \ref{robust.const} so that it uses a strong AMD code in place of a weak AMD code, then the resulting threshold scheme is $\epsilon$-\emph{robust} for any probability distribution defined on the set of secrets $\{0, \dots , m-1\}$.

\subsection{Organization of the paper}

The rest of the paper is organized as follows. In Section \ref{sec.oldnmss}, we review existing definitions of non-malleable secret sharing. Section \ref{sec.newnmss} introduces our new approach, which is based on preventing specific types of modification of the secret (by an adversary) that depend on a certain relation. In Section \ref{circ.sec}, we define new modifications of algebraic manipulation (AMD) codes that are appropriate for constructing non-malleable secret sharing schemes. We also show how optimal AMD codes of the desired type can be constructed from new types of difference families that we term ``circular external difference families''. Some variations are studied in Sections \ref{general.sec} and \ref{strong.sec}. Finally, Section \ref{summary.sec} is a brief summary and conclusion.

\section{Non-malleable secret sharing}\label{sec.oldnmss}

In this section, we review various definitions that have been proposed for non-malleable secret sharing. However, we should note that,
in general, non-malleability is a weaker notion than robustness. Robustness protects against arbitrary modifications of up to $k-1$ shares in a $(k,n)$-threshold scheme. On the other hand, non-malleability only guards against certain specified types of share modifications.

As far as we know, the first mention of non-malleable secret sharing in the literature was in a 2006 PhD thesis by Kenthapadi \cite{Kenthapadi} and in a corresponding paper on distributed noise generation \cite{EC:DKMMN06}. The authors refer to \emph{non-malleable verifiable secret sharing} as an extension of verifiable secret sharing (VSS). Although they do not provide a formal definition or construction of non-malleable VSS, the papers provide a brief, high-level definition. The authors state that ``a non-malleable VSS scheme ensures that the values shared by a non-faulty processor are completely independent of the values shared by the other processors; even exact copying is prevented." 


Another approach was suggested in the extended version \cite{C:IshPraSah10} of a 2008 paper on secure multi-party computation \cite{C:IshPraSah08}. This paper defines a $(2,2)$-threshold scheme to be  \emph{non-malleable} if an adversary cannot win the following game.

\begin{definition}[The IPS Malleability Game] \quad \vspace{-.2in}\label{nm.defn}\\
\begin{description}
\item[Step 1.] In a $(2,2)$-threshold scheme, two shares, $a$ and $b$, are generated for some secret $s$.
\item[Step 2.] The adversary modifies a single share, say  $a \rightarrow \tilde{a}$. 
\item[Step 3.] The reconstruction algorithm takes as input $\tilde{a}$ and $b$. The adversary wins if $\tilde{a} \neq a$ and the reconstruction algorithm outputs some valid secret $s'$.
\end{description}
\end{definition}

If the adversary can win the above game with probability at most $\epsilon$, they say that the scheme is $\epsilon$-non-malleable. Later work studying fairness in secure computation used the same definition of non-malleable secret sharing \cite{C:BenKum14, Gordon, TCC:GIMOS10}, as did a paper on universal composability \cite{C:Rosulek12b}. It has been noted that this definition of non-malleability can be achieved using AMD codes in the same way that they had previously been used to provide robustness \cite{EC:CDFPW08, Gordon}. 

\begin{remark}
Definition \ref{nm.defn} is very similar to the definition we provided for robust secret sharing (Definition \ref{robust.defn}). Aside from the fact that Definition \ref{nm.defn} is restricted to the case of $(2,2)$-threshold schemes (although this is easily generalized), the only remaining difference is that the IPS Malleability game checks if the inputted share $\tilde{a} \neq a$ rather than checking whether the reconstructed secret $s' \neq s$. This is a slightly stronger definition, given that if the reconstructed secret $s' \neq s$, then this implies that the inputted shares must differ. On the other hand, different shares do not necessarily imply a different reconstructed secret. 
\end{remark}

A third line of work began in 2010 with the introduction of non-malleable codes \cite{ICS:DPW10, ACM:DPW18}. In coding theory, non-malleable codes can be considered as a relaxation of error-correcting and error-detecting codes. The previously discussed AMD codes are an example of error-detecting codes. The definition of non-malleability involves families of \emph{tampering functions}. The goal is that tampering results in a decoded message that either is correct or ``independent of and unrelated to'' the original message. 

Goyal and Kumar initiated a formal study of non-malleable secret sharing \cite{STOC:GoyKum18}, where they presented non-malleable threshold schemes based on non-malleable codes. They provide formal definitions for non-malleable secret sharing that are generalized versions of non-malleable codes in the so-called ``split-state model''. Most, but not all, of the recent work that refers to non-malleable secret sharing is based on these definitions. The initial study focused on threshold schemes while later work provided constructions for general access structures \cite{C:ADNOPRS19, EC:BadSri19, C:GoyKum18}. These constructions 
satisfy statistical privacy and statistical non-malleability. A few followup works considered these definitions in the computational setting \cite{TCC:BriFaoVen19, C:FaoVen19}. 

Finally, there is a definition of non-malleable secret sharing in the unconditionally secure setting \cite{batchdpir} based on the ideas from Goyal and Kumar \cite{STOC:GoyKum18}; however, this definition is specific to \emph{incremental} secret sharing schemes. In an incremental scheme, a running tally is maintained as shares are submitted and the final tally corresponds to the original secret. This definition (and corresponding construction) does not extend to the general setting because, for example, their game-based definition precludes the possibility of the last share being controlled by the adversary.

\section{A new approach to non-malleable secret sharing}\label{sec.newnmss}

We propose a new game-based definition of non-malleability, based on a specified binary 
relation $\sim$ on the set $\mathcal{S}$ of possible secrets (thus we can regard $\sim$ as a subset of $\mathcal{S} \times \mathcal{S}$).
We are only interested in \emph{irreflexive} relations (i.e., relations in which $s \sim s$ never holds) in this paper.

The basic idea is that the adversary's goal is to modify one or more shares in such a way that 
$s' \sim s$, where $s$ is the true secret and $s' \neq s$ is the reconstructed secret. 
With this approach, we can provide more fine-grained constructions of non-malleable schemes, which are secure with respect to specific relations. 


\begin{definition}[The $\sim$-Malleability Game] 
Assume a $(k,n)$-threshold scheme. The secret is $s$. Fix some $t$, $1 \leq t < k$, and let $\sim$ be a specified binary relation over the set of $\mathcal{S}$ of possible valid secrets. 
\begin{description}
\item[Step 1.] The dealer chooses a secret $s \in \mathcal{S}$ and constructs $n$ valid shares for the secret $s$.
\item[Step 2.] $t$ of the $n$ shares are given to the adversary.
 The adversary then modifies the $t$ shares to create new ``bad shares".
\item[Step 3.] A secret $s'$ is reconstructed using $k$ shares chosen by the adversary, $t$ of which are the ``bad shares". The good shares used during reconstruction may be chosen (but not modified) by the adversary; however, their values are not known to the adversary.
\end{description}
The adversary wins the $\sim$-malleability game if the reconstructed secret $s'$ is a valid secret such that $s' \neq s$ and $s' \sim s$.

If the adversary cannot win the non-malleability game with probability greater than $\epsilon$, then we say that the scheme is an
\emph{$\epsilon$-secure $t$-non-malleable} threshold scheme with respect to the binary relation $\sim$.
\end{definition}

We note that if we define the relation $\sim$ by $a' \sim a$ if and only if $a \neq a'$, then the only requirement for the adversary to win the malleability game  is that $s' \neq s$. Then, this definition is equivalent to the definition of robust secret sharing. It follows from this observation that if a $(k,n)$-threshold scheme is robust, it is also non-malleable with respect to the relation $\neq$.

\begin{remark}
Our definition of non-malleable secret sharing is motivated by considerations raised in the initial study of non-malleable cryptography given in \cite{SJC:DolDwoNao00}, where it is stated that ``given the ciphertext, it is impossible to generate a different ciphertext so that the respective plaintexts are related.''
\end{remark}

We will pay particular attention to ``additive'' relations, which we define now. 

\begin{definition}
Suppose that $m$ is a fixed positive integer, and let $0 < c < m$. 
Define the relation $\sim_{c}$ as follows: $s' \sim s$ if and only if $s' = s+c \bmod m$.
\end{definition}

We now prove that a $(k,n)$-Shamir threshold scheme is not non-malleable, i.e., it is malleable, for a  relation $\sim_c$ of this type.

\begin{theorem}
Suppose $p$ is prime and $0 < c < p$. Then a  $(k,n)$-Shamir threshold scheme with secrets in $\zed_p$ is malleable with respect to the relation $\sim_c$.
\end{theorem}
\begin{proof}
We describe how the adversary can carry out a successful attack by modifying only a single share. Suppose the shares $v_1,\dots,v_k$ are used to reconstruct a polynomial. 
Assume that the adversary modifies a single share $v_1$ by adding some nonzero value $\delta$. That is, upon reconstruction, the adversary submits $v_1' = v_1 + \delta \bmod p$.
Then, as was shown in the  the proof of Theorem \ref{thm.tw1}, the reconstructed secret will be
\[
s' = b_1v_1' + \sum_{i=2}^k b_iv_i = s + b_1\delta \bmod p,
\]
where $b_1$ is a (known) Lagrange coefficient.
If the adversary  chooses $\delta = c(b_1)^{-1} \bmod p$, then $s' \sim_c s$, as desired, and the adversary wins the $\sim_c$-malleability game.
\end{proof}

\section{Circular external difference families}
\label{circ.sec}
In this section we construct unconditionally secure threshold schemes that are non-malleable with respect to a relation $\sim_c$. Our construction is based on an appropriate modification of AMD codes. 
Suppose we modify \Cref{def.amd} as follows:

\begin{definition}[Circular weak AMD code]
Let $\mathcal{G}$ be an additive abelian group of order $n$ and let $\mathcal{A} = \{{A}_0,\dots,{A}_{m-1}\}$ be $m$ pairwise disjoint $\ell$-subsets of $\mathcal{G}$. Let $0 < \epsilon < 1$ and let $c$ be a fixed integer such that $1 \leq c \leq m-1$. Then $(\mathcal{G},\mathcal{A})$ is an \emph{$\epsilon$-secure $c$-circular $(n,m,\ell)$-AMD code} if an adversary cannot win the following \emph{circular AMD game} with  probability greater than $\epsilon$.
\begin{enumerate}
\item The adversary chooses a value $\Delta \in \mathcal{G} \setminus \{0\}$. 
\item The \emph{source} $i \in \{0,\dots,m-1\}$ is chosen uniformly at random.
\item The source is encoded by choosing $g$ uniformly at random from ${A}_i$.
\item The adversary wins if and only if $g + \Delta \in A_j$ where $j = i + c \bmod m$. 
\end{enumerate}
\end{definition}

The only difference between this definition and \Cref{def.amd} is the additional requirement in part 4 that $j = i + c \bmod m$. If we can construct such a code, then we will immediately obtain a threshold scheme that is non-malleable with respect to the relation $\sim_c$.

Previous work has studied the connection between AMD codes and external difference families (e.g., see \cite{DM:PatSti16}). We follow suit by defining a modification of external difference families that will yield (optimal) circular AMD codes. First, we define some notation. Let $\mathcal{G}$ be an abelian group. For any two disjoint sets $A_1,A_2 \subseteq \mathcal{G}$, define
\[
\mathcal{D}(A_1,A_2) = \{x-y : x \in A_1,y\in A_2\}.
\]
Note that $\mathcal{D}(A_1,A_2)$ is a multiset.

\begin{definition}[Circular external difference family (CEDF)]
Let $G$ be an additive abelian group of order $n$. Suppose $m \geq 2$ and $1 \leq c \leq m-1$. 
An $(n, m, \ell; \lambda)$-$c$-circular external difference family (or $(n, m, \ell; \lambda)$-$c$-CEDF) is a set of $m$ disjoint $\ell$-subsets of $G$, say $\mathcal{A} = (A_0,\dots,A_{m-1})$, such that the following multiset equation holds:
\[
\bigcup_{j=0}^{m-1} \mathcal{D}(A_{j+c\bmod m}, A_j) = \lambda (G \setminus \{0\}).
\]
We observe that $m \ell^2 = \lambda (n-1)$ if an $(n, m, \ell; \lambda)$-$c$-CEDF exists.
\end{definition}

\begin{example}
\label{E13-3-2-1}
The following three sets of size $2$ form a $(13, 3, 2, 1)$-$1$-CEDF in $\mathbb{Z}_{13}$:
\[
\mathcal{A} = (\{1,12\},\{4,9\},\{3,10\}).
\]
This is easily verified from the following computations:
\begin{eqnarray*}
\mathcal{D}(A_{1}, A_0) & = & \{ 3,5,8,10\}\\
\mathcal{D}(A_{2}, A_1) & = & \{ 12,7,6,1\}\\
\mathcal{D}(A_{0}, A_2) & = & \{ 11,4,9,2\}.
\end{eqnarray*}
\end{example}

Paterson and Stinson \cite{DM:PatSti16} define  R-optimal (weak) AMD codes and prove their equivalence to external difference families. A weak AMD code is \emph{R-optimal} if the adversary's probability of winning the game in \Cref{def.amd} is minimized.

\begin{definition}[R-optimal weak AMD code]
A weak $(n,m,\ell)$-AMD code is \emph{R-optimal} if it is $\epsilon$-secure with 
\[
\epsilon = \frac{\ell(m-1)}{n-1}.
\]
\end{definition}

In the case of weak circular AMD codes, the appropriate definition of $R$-optimality is the following.

\begin{definition}[R-optimal weak circular AMD code]
A weak $c$-circular $(n,m,\ell)$-AMD code is \emph{R-optimal} if it is $\epsilon$-secure with 
\[
\epsilon = \frac{\ell}{n-1}.
\]
\end{definition}

The following theorem is proven using an argument similar to the proof of  \cite[Theorem~3.10]{DM:PatSti16}, 
which established the equivalence of external difference families and optimal AMD codes.

\begin{theorem}\label{thm.camdtocedf}
An R-optimal $c$-circular weak $(n,m,\ell)$-AMD code is equivalent to an $(n,m,\ell;\lambda)$-$c$-CEDF.
\end{theorem}

\begin{proof}
Suppose $\mathcal{A} = \{{A}_0,\dots,{A}_{m-1}\}$ is a $c$-circular weak $(n,m,\ell)$-AMD code over an abelian group $\mathcal{G}$.
Define \[\mathcal{T} = \{ (i,x,g): 0 \leq i \leq m-1, x \in A_i, g \in G \setminus \{0\}, x+g \in A_{i + c \bmod m}\}.\] We observe that
$|\mathcal{T}| = m \ell^2$. For any fixed value $\Delta \in G \setminus \{0\}$, define \[\mathcal{T}_{\Delta} = 
\{  (i,x,\Delta) \in \mathcal{T}  \}.\] 

Suppose the adversary  chooses a particular value $\Delta$ in the AMD game. 
The adversary wins the AMD game if $(i,x,\Delta) \in \mathcal{T}_{\Delta}$. 
There are  ${m \ell}$ pairs $(i,x)$ such that $(i,x,g) \in \mathcal{T}$ for some $g$ and each of these pairs is equally likely.
 Hence, for a given value of $\Delta$, the success probability of the adversary in winning the AMD game is 
\[ \epsilon_{\Delta} = \frac{|\mathcal{T}_{\Delta}|}{m \ell}.\]

Now, since 
\[ \sum_{\Delta \in G \setminus \{0\}} |\mathcal{T}_{\Delta}| = |\mathcal{T}| = m \ell^2,\]
it follows that there exists $\Delta \in G \setminus \{0\}$ such that 
\[  |\mathcal{T}_{\Delta}| \geq \frac{m \ell ^2}{n-1},\]
and hence \[ \max_{\Delta} \{ \epsilon_{\Delta}\} \geq \frac{\ell}{n-1}.\]

We  further observe that \[\max_{\Delta} \{ \epsilon_{\Delta}\} = \frac{\ell}{n-1}\]
if and only if \[|\mathcal{T}_{\Delta}| = \frac{m \ell ^2}{n-1}\] for every $\Delta \in G \setminus \{0\}$.
But this says that the AMD code is $(n,m,\ell;\lambda)$-CEDF with $\lambda = m \ell^2 / (n-1)$.

The converse is straightforward; we leave the details for the reader to verify.
\end{proof}

Suppose we fix the value $c$ to be $c=1$. 
In order to construct an optimal circular AMD code (and the corresponding non-malleable threshold scheme with respect to $\sim_1$), we need only construct an $(n,m,\ell;\lambda)$-$1$-CEDF. We present some simple examples next.

\begin{constructionx}
\label{const1}
Suppose $q = m\ell^2+1$ is a prime power. There is a natural way to try to construct a $(q,m,\ell,1)$-$1$-CEDF based on taking the unique subgroup of
${\eff_q}^*$ of order $\ell$ along with some cosets. Let $\alpha \in \eff_q$ be a primitive element. 
Define
\[ C_0 = \{  \alpha^{i\ell m} : 0 \leq i \leq \ell-1\},\] and for $1 \leq j \leq m-1$, define
\[ C_j = \alpha^{\ell j} C_0,\]
where all arithmetic is in $\eff_q$. Finally, let $\A = (C_0, \dots , C_{m-1})$. 
\end{constructionx}

Construction \ref{const1} may or may not yield a $(q,m,\ell,1)$-$1$-CEDF.
However,  Example \ref{E13-3-2-1} is obtained by this method. 
It frequently happens that this construction yields a $(q,m,\ell,1)$-$1$-CEDF. Here are a few more small examples
obtained from Construction \ref{const1} using groups of  prime order.

\begin{example} 
\label{ex1742}
{\rm 
The following four sets of size $2$ form a $(17,4,2,1)$-$1$-CEDF in $\zed_{17}$:
\[\A = (
\{1 ,16\},
\{9 ,8\},
\{13 ,4\},
\{15 ,2)\}).\]}
\end{example}

\begin{example} 
\label{ex15165}
{\rm 
The following six sets of size $5$ form a $(151,6,5,1)$-$1$-CEDF in $\zed_{151}$:
\begin{eqnarray*}
\A &=& (
\{1 ,59 ,8 ,19 ,64\},
\{75 ,46, 147, 66 ,119\},
\{38 ,128, 2 ,118, 16\},\\ &&
\{132, 87, 150 ,92 ,143\},
\{85 ,32 ,76, 105 ,4\},
\{33 ,135, 113 ,23 ,149\}
).
\end{eqnarray*}
}
\end{example}

\begin{example} 
\label{ex2972}
{\rm 
The following seven sets of size $2$ form a $(29,7,2,1)$-$1$-CEDF in $\zed_{29}$:
\[\A =(
\{1, 28\},
\{4,25\},
\{16,13\},
\{6,23\},
\{24,5\},
\{9,20\},
\{7,22\}
).\]
}
\end{example}

\begin{example} 
\label{ex7383}
{\rm 
The following eight sets of size $3$ form a $(73,8,3,1)$-$1$-CEDF in $\zed_{73}$:
\begin{eqnarray*}
\A &=& (
\{1, 8 ,64\},
\{52, 51, 43\},
\{3 ,24, 46\},
\{10, 7 ,56\},\\ &&
\{9 ,72, 65\},
\{30, 21, 22\},
\{27 ,70 ,49\},
\{17, 63, 66\}
).
\end{eqnarray*}
}
\end{example}

In the case $\ell=2$, we have a simple condition that determines when Construction \ref{const1} yields a $(q,m,2,1)$-$1$-CEDF in $\eff_q$.

\begin{theorem}
\label{thm1}
Suppose that $q = 4m+1$ is a prime power and $\alpha$ is a primitive element of $\eff_q$. Then 
Construction \ref{const1} with $\ell = 2$ yields a $(q,m,2,1)$-$1$-CEDF in $\eff_q$ if and only if $\alpha^4-1$ is a quadratic non-residue in ${\eff_q}^*$.
\end{theorem}

\begin{proof}
Here we have $C_0 = \{1,-1\}$ and $C_j= \{ \alpha^{2j}, -\alpha^{2j}\}$, for $1 \leq j \leq m-1$.
Therefore $\D (C_1,C_0) = \{\pm (\alpha^2-1), \pm (\alpha^2+1)\}$. More generally,
\[\D (C_{j+1 \bmod m},C_j) = \{\pm \alpha^{2j}(\alpha^2-1), \pm \alpha^{2j}(\alpha^2+1)\}\] for $0 \leq j \leq m-1$.
Let $C = \{ \pm \alpha^{2j}: 0 \leq j \leq m-1\}$. Note that $C$ consists of the quadratic residues in $\eff_q$.
We have
\[ \bigcup _{j=0 }^{m-1}\D(C_{j+1 \bmod m},C_j)  = \{ \alpha^2-1, \alpha^2+1 \} \otimes C,\]
where, for two sets $X, Y \subseteq \zed_p$, we define
$X \otimes Y = \{xy: x \in X, y \in Y\}$.  Clearly, $\{ \alpha^2-1, \alpha^2+1 \} \otimes C = {\eff_q}^*$ if and only if one of 
$\alpha^2-1, \alpha^2+1$ is a quadratic residue and the other is a quadratic nonresidue. This is equivalent to the condition
that $\alpha^4-1$ is a quadratic non-residue in ${\eff_q}^*$.
\end{proof}

\begin{remark}
Steven Wang observed that results of Cohen, Sharma and Sharma \cite{cohen} can be used to prove, for all sufficiently large prime powers $q$, that a primitive element $\alpha \in \eff_q$ exists such that $\alpha^4-1$ is a quadratic nonresidue. Indeed, \cite{cohen} shows that there is a constant $c$ such that, for all prime powers $q > c$, there is a primitive element $\alpha \in \eff_q$ such that $\alpha^4-1$ is also primitive. Of course, for $q \equiv 1 \bmod 4$, a primitive element in $\eff_q$ must be a quadratic nonresidue. The paper \cite{cohen} also shows that $c \approx 7.867 \times 10^8$.
\end{remark}

For arbitrary $\ell \geq 2$, we have the following more general result. 

\begin{theorem}
\label{thm2}
Suppose that $q = m\ell^2 +1$ is a prime power and $\alpha$ is a primitive element of ${\eff_q}$. Define $\beta = \alpha^{\ell}$
and let $H$ be the subgroup of ${\eff_q}^*$ of order $\ell m$ generated by $\beta$.
Then Construction \ref{const1} yields a $(q,m,\ell,1)$-$1$-CEDF in $\eff_q$ if and only if 
\[\{ \beta - 1 , \beta^{m+1}-1, \dots , \beta^{(\ell-1)m+1} -1 \}\] is
a set of coset representatives of $H$ in ${\eff_q}^*$.
\end{theorem}

\begin{proof}
We have $C_0 = \{  \alpha^{i\ell m} : 0 \leq i \leq \ell -1\}$. 
Since $\beta = \alpha^{\ell}$, we have 
\[C_j = \{  \alpha^{i\ell m+\ell j} : 0 \leq i \leq \ell - 1\} = \{  \beta^{im+j} : 0 \leq i \leq \ell - 1\} ,\]
for $0 \leq j \leq m-1$.
Hence,
\begin{eqnarray*}
\D (C_{j+1 \bmod m},C_j) &=& \{ \beta^{i_1m+j+1} - \beta^{i_2m+j} : 0 \leq i_1,i_2 \leq \ell - 1\}\\
&=& \{ \beta^j ( \beta^{i_1m+1} - \beta^{i_2m} ) : 0 \leq i_1,i_2 \leq \ell - 1\}
\end{eqnarray*}
 for $0 \leq j \leq m-1$.

Let \[B = \{ \beta^{j}: 0 \leq j \leq m-1\}\] and 
\[E = \{  \beta^{i_1m+1} - \beta^{i_2m}  : 0 \leq i_1,i_2 \leq \ell - 1\}.\]
We have
\[ \bigcup _{j=0 }^{m-1}\D(C_{j+1 \bmod m},C_j)  = E \otimes B.\]
We can further express
\[ E = \{ \beta - 1 , \beta^{m+1}-1, \dots , \beta^{(\ell -1)m+1} -1 \} \otimes \{1, \beta^m , \dots , \beta^{(\ell -1)m} \}.\]
It is easily seen that 
\[B \otimes \{1, \beta^m , \dots , \beta^{(\ell -1)m} \} = \{ \beta ^ i : 0 \leq i \leq \ell m-1 \} = H.\] 
Thus we have 
\[ E \otimes B = \{ \beta - 1 , \beta^{m+1}-1, \dots , \beta^{(\ell -1)m+1} -1\} \otimes H .\]
Hence, it follows that 
Construction \ref{const1} yields a $(q,m,\ell,1)$-CEDF in $\eff_q$ if and only if 
$\{ \beta - 1 , \beta^{m+1}-1, \dots , \beta^{(\ell -1)m+1} -1\}$ is
a set of coset representatives of $H$ in ${\eff_q}^*$.
\end{proof}

\begin{remark}
If we take $\ell=2$ in Theorem \ref{thm2}, then the resulting condition is equivalent to that of Theorem \ref{thm1}.
\end{remark}

Table \ref{tab1} lists parameters of CEDFs that we obtained from Construction \ref{const1} in the cases where $p= \ell^2m+1$ is prime.
We considered all parameter sets with $m \leq 50$ and  $\ell \leq 10$. Each quadruple in Table \ref{tab1} has the form
$(p,m,\ell,\alpha)$, where $\alpha$ is the primitive root used in Construction \ref{const1}.

\begin{table}
\caption{$(p,m,\ell;1)$-$1$-CEDFs obtained from Construction \ref{const1}}
\label{tab1}
\[
\begin{array}{|l|l|l|l|}
\hline
(13,3,2,2) & (17,4, 2, 3) & (151,6, 5, 6) & (29,7, 2, 2)\\ \hline
(73, 8, 3, 5) & (37, 9, 2, 2) & (41, 10, 2, 6) & (53, 13, 2, 8) \\ \hline
(127, 14, 3, 116) & (61, 15, 2, 35) & (241, 15, 4, 7) & (401, 16, 5, 27) \\ \hline
(73, 18, 2, 5) & (1217, 19, 8, 642) & (181, 20, 3, 57) & (337, 21, 4, 10) \\ \hline
(757, 21, 6, 2) & (89, 22, 2, 51) & (199, 22, 3, 44) & (97, 24, 2, 5)\\ \hline
(101, 25, 2, 2) & (401, 25, 4, 3) & (109, 27, 2, 6) & (433, 27, 4, 94) \\ \hline
(113, 28, 2, 3) & (271, 30, 3, 142) & (137, 34, 2, 3) & (307, 34, 3, 241) \\ \hline
(577, 36, 4, 230) & (149, 37, 2, 2) & ( 593, 37, 4, 339 ) & (157, 39, 2, 142) \\ \hline
(641, 40, 4, 264) & (379, 42, 3, 233) & (673, 42, 4, 5) & (173, 43, 2, 128) \\ \hline
(1549, 43, 6, 1165) & (397, 44, 3, 296) & ( 181, 45, 2, 28 ) & (193, 48, 2, 5) \\ \hline
(433, 48, 3, 393) & (769, 48, 4, 453) & (197, 49, 2, 32) & \\ \hline
\end{array}
\]
\end{table}



We record a couple of simple results concerning $(n, m, \ell; \lambda)$-$c$-CEDF with $c > 1$.

\begin{theorem} Suppose there is an $(n, m, \ell; \lambda)$-$1$-CEDF and suppose $\gcd(c,m) = 1$, where $1 < c < m-1$. 
Then there is an  $(n, m, \ell; \lambda)$-$c$-CEDF.
\end{theorem}

\begin{proof}
Suppose $\mathcal{A} = (A_0,\dots,A_{m-1})$ is an $(n, m, \ell; \lambda)$-$1$-CEDF. For $0 \leq j \leq m-1$, define
$A'_i = A_{i c^{-1} \bmod n}$. Then $\mathcal{A}' = (A'_0,\dots,A'_{m-1})$ is the desired $(n, m, \ell; \lambda)$-$c$-CEDF.
\end{proof}

\begin{theorem} Any $(n, m, \ell; \lambda)$-$1$-CEDF is  
an $(n, m, \ell; \lambda)$-$(m-1)$-CEDF.
\end{theorem}

\begin{proof}
Immediate.
\end{proof}

\subsection{More general relations}
\label{general.sec}

Our main construction, Construction \ref{const1}, applies to the relation $\sim_1$, where $b \sim_1 a$ if $b = a+1$. The relation $\sim_1$ can be thought of as a directed cycle of length $m$. This is an example of a (directed) Cayley graph in $\zed_m$. 
More generally, we can consider a subset $S \subseteq \zed_m \setminus \{0\}$ and define a  relation $\sim_S$  as follows:
$b \sim_S a$ if $b-a \in S$. This relation corresponds to the directed Cayley graph that is often denoted as $\Gamma(\zed_m,S)$.

We can define $S$-external difference families in the obvious way for an arbitrary abelian group $G$.

\begin{definition}[$S$-external difference family, or $S$-EDF]
Let $G$ be an additive abelian group of order $n$. Suppose $m \geq 2$ and suppose $S \subseteq G\setminus \{0\}$. 
An $(n, m, \ell; \lambda)$-$S$-external difference family (or $(n, m, \ell; \lambda)$-$S$-EDF) is a set of $m$ disjoint $\ell$-subsets of $G$, say $\mathcal{A} = (A_0,\dots,A_{m-1})$, such that the following multiset equation holds:
\[
\bigcup _{c \in S} \bigcup_{j=0}^{m-1} \mathcal{D}(A_{j+c\bmod m}, A_j) = \lambda (G \setminus \{0\}).
\]
We observe that $|S| m \ell^2 = \lambda (n-1)$ if an $(n, m, \ell; \lambda)$-$S$-EDF exists.
\end{definition}

To illustrate, we present some examples of $S$-EDF for $S = \{1,2\}$. These are constructed using the basic idea of Construction \ref{const1}. Note that these $S$-EDF have $\lambda = |S| = 2$ and $n = m \ell^2 + 1$.

\begin{example} 
\label{ex2972-2}
{\rm 
The following seven sets of size $2$ form a $(29,7,2,2)$-$\{1,2\}$-CEDF in $\zed_{29}$:
\[\A =(
\{1, 28\},
\{9,20\},
\{23,6\},
\{4,25\},
\{7,22\},
\{5,24\},
\{13,16\}
).\]
}
\end{example}

\begin{example} 
\label{ex3743}
{\rm 
The following four sets of size $3$ form a $(37,4,3,2)$-$\{1,2\}$-CEDF in $\zed_{37}$:
\[\A =(
\{1, 26, 10\},
\{8, 23, 6\},
\{27, 36, 11\},
\{31, 29, 14\}
).\]
}
\end{example}

\begin{example} 
\label{ex11374}
{\rm 
The following seven sets of size $4$ form a $(113,7,4,2)$-$\{1,2\}$-CEDF in $\zed_{113}$:
\begin{eqnarray*}
\A &=& (
\{1, 98, 112, 15\},
\{81, 28, 32, 85\},
\{7, 8, 106, 105,\},
\{2, 83, 111, 30\}\\ &&
\{49, 56, 64, 57\},
\{14, 16, 99, 97\},
\{4, 53, 109, 60\}
).
\end{eqnarray*}
}
\end{example}


\subsection{Strong circular AMD codes}
\label{strong.sec}

We can also distinguish between strong and weak circular AMD codes. Thus far, we have presented weak versions of circular AMD codes and CEDFs. These are secure in settings where the source (i.e., the secret) is assumed to be uniformly distributed. Related definitions of strong circular AMD codes and strong CEDFs would be relevant in settings where the secret is not uniformly distributed. These definitions are presented next.

\begin{definition}[Strong circular AMD code]
Let $\mathcal{G}$ be an additive abelian group of order $n$ and $\mathcal{A} = \{ {A}_0,\dots, {A}_{m-1}\}$ be $m$ pairwise disjoint $\ell$-subsets of $\mathcal{G}$. Let $0 < c \leq m-1$ and let $0 < \epsilon < 1$. Then $(\mathcal{G},\mathcal{A})$ is an  \emph{ $\epsilon$-secure strong $c$-circular $(n,m,\ell)$-AMD code} if an adversary cannot win the following \emph{strong $c$-circular AMD game} with probability greater than $\epsilon$.
\begin{enumerate}
\item The \emph{source} $i \in \{0,\dots,m-1\}$ is specified and given to the adversary.
\item The adversary chooses a value $\Delta \in \mathcal{G} \setminus \{0\}$. 
\item The source is encoded by choosing $g$ uniformly at random from ${A}_i$.
\item The adversary wins if and only if $g + \Delta \in A_j$ for $j = i + c \bmod m$. 
\end{enumerate}
\end{definition}

\begin{definition}[Strong circular external difference family (SCEDF)]
Let $G$ be an additive abelian group of order $n$. An $(n, m, \ell; \lambda)$-strong $c$-circular external difference family (or $(n, m, \ell; \lambda)$-$c$-SCEDF) is a set of $m$ disjoint $\ell$-subsets of $G$, say $\mathcal{A} = (A_0,\dots,A_{m-1})$, such that the following multiset equation holds for every $j$, $0 \leq j \leq m-1$:
\[
\mathcal{D}(A_{j+c\bmod m}, A_j) = \lambda (G \setminus \{0\}).
\]
\end{definition}

In general, SCEDF seem to be difficult to construct. Of course there are trivial examples, since any $(n, m, \ell; \lambda)$-$1$-CEDF with $\ell=2$ is automatically strong. However, at present, we are unable to construct any $(n, m, \ell; \lambda)$-$c$-SCEDF with $\ell \geq 3$. This is an interesting open problem.

\section{Conclusion}
\label{summary.sec}


We have proposed a new definition for non-malleable secret sharing schemes and discussed how certain types of these schemes can be obtained from circular external difference families. There are many problems that can be considered in future work. Here are a few such problems.

\begin{itemize}
\item Determine whether any nontrivial examples of strong circular difference families exist.
\item Construct infinite classes of (weak) circular difference families, especially when $\lambda = 1$.
\item Circular difference families yield optimal circular weak AMD codes. However, there are many situations where desired optimal circular difference families do not exist. In these cases, it would be useful to have ``good'' non-optimal constructions for AMD codes.
\item Find additional constructions of $S$-external difference families, for various choices of subsets $S \subseteq \zed_m$. 
\item Determine if Theorem \ref{thm1} can be used to construct a $(q,m,2,1)$-$1$-CEDF in $\eff_q$ for all prime powers $q = 4m+1$.
\end{itemize}

\section*{Acknowledgements}
We thank Steven Wang for bringing the results of \cite{cohen} to our attention.


\end{document}